\documentclass{fairmeta_gr2}

\usepackage{amsmath}
\usepackage{amssymb}
\usepackage{amsfonts}
\usepackage{amsthm}
\usepackage{pifont}
\usepackage{enumitem}
\usepackage[normalem]{ulem}

\microtypesetup{expansion=false}

\newtheorem{proposition}{Proposition}

\crefname{proposition}{Proposition}{Propositions}
\Crefname{proposition}{Proposition}{Propositions}
\crefname{corollary}{Corollary}{Corollaries}
\Crefname{corollary}{Corollary}{Corollaries}

\definecolor{accGreen}{RGB}{180,40,40}
\definecolor{accBlue}{RGB}{40,90,180}
\newcommand{\accBetter}[1]{\,\textcolor{accGreen}{$_{+#1}$}}
\newcommand{\accWorse}[1]{\,\textcolor{accBlue}{$_{-#1}$}}
\newcommand{\spdBetter}[1]{\,\textcolor{accGreen}{$_{#1\times}$}}
\newcommand{\spdWorse}[1]{\,\textcolor{accBlue}{$_{#1\times}$}}

\title{Bifocal Diffusion Language Models: Asymmetric Bidirectional Context for Parallel Generation}

\author[2,*]{Yuhang Chen}
\author[1]{Xianfeng Wu}
\author[2]{Jinhao Duan}
\author[1]{Mingfu Liang}
\author[1]{Xiaohan Wei}
\author[1]{Yunchen Pu}
\author[1]{Fei Tian}
\author[1]{Chonglin Sun}
\author[1]{Parish Aggarwal}
\author[1]{Frank Shyu}
\author[1]{Luke Simon}
\author[1]{Sandeep Pandey}
\author[1,\dagger]{Xi Liu}
\author[2,\dagger]{Tianlong Chen}

\affiliation[1]{Meta AI}
\affiliation[2]{University of North Carolina at Chapel Hill}

\contribution[*]{Work done during an internship at Meta.}
\contribution[\dagger]{Joint corresponding authors.}

\correspondence{Xi Liu (\email{xliu1@meta.com}) and Tianlong Chen (\email{tianlong@cs.unc.edu})}

\abstract{Discrete diffusion language models (dLLMs) recover masked tokens in parallel, offering significant speedups over autoregressive (AR) generation. However, such promising frameworks face a fundamental architectural design dilemma: \ding{182} Adopting bidirectional attention achieves strong generation quality by allowing each position to access the full context, but is inherently incompatible with KV caching, limiting inference throughput in batch-serving scenarios; \ding{183} Conversely, causal attention enables efficient cached inference but loses all right-side context, substantially degrading generation quality. This paper introduces Bifocal dLLMs, a new paradigm that resolves this dilemma through \emph{asymmetric bidirectional context}. Analogous to bifocal lenses, we instantiate the paradigm as \textbf{R2LM} (Right-to-Left Mamba), which combines two complementary mechanisms: $a$) standard causal attention providing precise left-context with full KV cache compatibility, while $b$) a lightweight reverse Mamba SSM sidecar supplying compressed right-side context without breaking cacheability. Comprehensive experiments on continued pretraining of Qwen3-1.7B with 60B tokens demonstrate that R2LM achieves $2.4\times$ to $12.9\times$ higher throughput than bidirectional dLLMs and $1.9\times$ to $2.9\times$ speedup over AR baselines in batch serving through parallel decoding with KV caching, while exceeding the causal baseline on most benchmarks and surpassing the bidirectional dLLM on average.}

\date{\today}

\begin{document}

\maketitle

\section{Introduction}
\label{sec:intro}

Discrete diffusion language models (dLLMs)~\citep{austin2021structured, sahoo2024simple, shi2024simplified, ou2025absorbingdiscretediffusionslm} are a family of generative language models in which text is produced by iteratively denoising masked sequences. They have recently emerged as a promising alternative to the autoregressive (AR) paradigm~\citep{brown2020language, touvron2023llama, yang2024qwen2} dominant in language modeling, matching AR quality at scale~\citep{nie2025llada, ye2025dream, liu2025wedlm}. The appeal is throughput: AR generation requires $N$ sequential forward passes to produce $N$ tokens, regardless of available parallelism. dLLMs recover multiple tokens per step and produce the same $N$ tokens in $T \ll N$ passes, delivering $3$ to $8\times$ wall-clock speedups over optimized AR engines.

However, the dLLM paradigm itself faces an internal architectural dilemma. \ding{182} Most high-quality dLLMs~\citep{nie2025llada, ye2025dream} employ \textbf{bidirectional attention}, allowing every position to aggregate information from all others. It maximizes the context available, but each token's key-value (KV) representation depends on \emph{all} other tokens (including future, yet-to-be-resolved ones), so standard prefix KV caching is impossible and every denoising step must recompute the full attention. \ding{183} An alternative line of work adopts \textbf{causal attention}~\citep{arriola2025block, liu2025wedlm, cheng2025sdar}, preserving the standard left-to-right dependency structure and thereby enabling KV caching. But causal attention loses all right-side context.

\begin{figure*}[t]
\centering
\begin{minipage}[c]{0.68\textwidth}
\centering
\includegraphics[width=\linewidth]{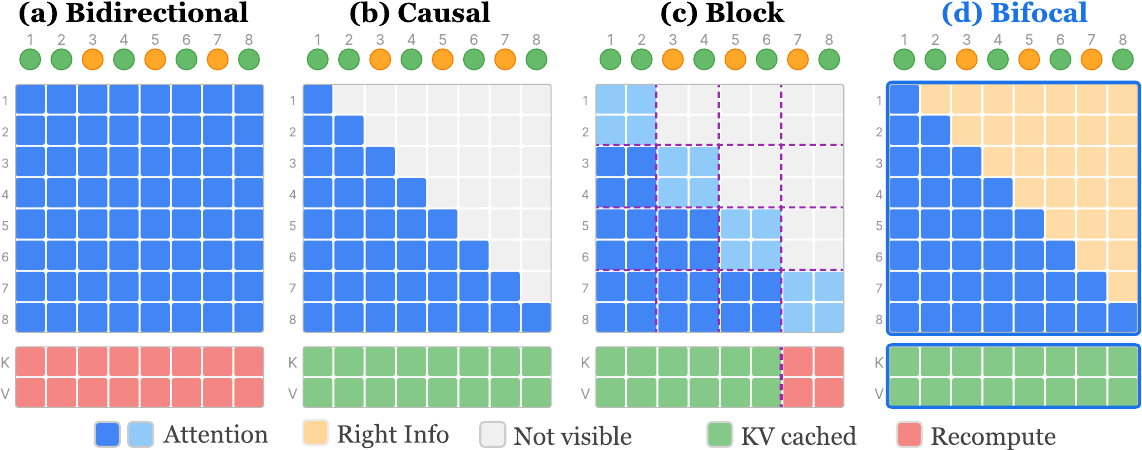}
\end{minipage}\hfill
\begin{minipage}[c]{0.30\textwidth}
\centering
\includegraphics[width=\linewidth]{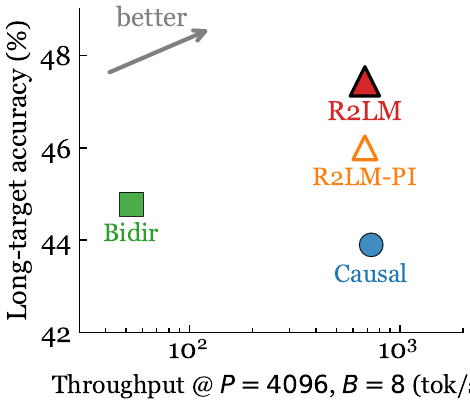}
\end{minipage}
\caption{\textbf{Left}: the asymmetric bidirectional context idea at a glance, contrasting bidirectional dLLMs (full attention, no KV cache), causal dLLMs (KV cache, no right context), and Bifocal (causal attention plus an R2LM sidecar that supplies right-side context while preserving cacheability). \textbf{Right}: quality versus throughput at $P{=}4096$, $B{=}8$ on a single H100. R2LM sits in the upper-right region (high quality, high throughput).}
\label{fig:teaser}
\end{figure*}

This dilemma has driven a range of compromise paradigms summarized in \cref{tab:paradigm}. Block-diffusion variants~\citep{arriola2025block, cheng2025sdar, wu2025fastdllmv2} restrict bidirectional attention to fixed-size blocks, recovering partial KV caching at the cost of bounding right-side context to a single block. Purely causal variants~\citep{liu2025wedlm, ruan2026card} introduce topological reordering or specialized masking schedules to maximize prefix cacheability, but accept the right-context quality gap. Hybrid AR-diffusion variants~\citep{refusion2026, liu2025tidar, sahoo2025esolms} combine autoregressive and diffusion modes through dual-mode architectures or slot-level reordering, at the cost of architectural complexity. Across all cases, the central tension persists: \emph{causal attention restricts each position to left-only context, and no amount of inference optimization can recover the missing right-side information.}
We thus ask:

\begin{tcolorbox}[before skip=0.2cm, after skip=0.2cm, boxsep=0.0cm, middle=0.1cm, top=0.1cm, bottom=0.1cm]
\textit{Can right-side context be supplied through a non-attention pathway, leaving the prefix KV cache intact?}
\end{tcolorbox}

We propose \textbf{Bifocal dLLMs}, a new dLLM paradigm that supplies right-side context through a \emph{non-attention pathway} in parallel with causal attention (\cref{fig:teaser}; rightmost row of \cref{tab:paradigm}). The underlying design principle, \emph{asymmetric bidirectional context}, treats the two directions of context flow through \emph{different mechanisms} rather than the symmetric handling of bidirectional attention---preserving the prefix KV cache that causal attention enables.

We instantiate this paradigm as \textbf{R2LM}: the right-context pathway is a reverse Mamba SSM~\citep{gu2024mamba} sidecar attached to a pretrained autoregressive backbone via forward hooks at every $k$-th decoder layer, with zero-initialized contribution so that the augmented model is bit-identical to the causal Transformer at initialization.

The R2LM layers are position-aware by construction: as a sequential state space model, Mamba processes tokens in order, maintaining the positional structure that distinguishes meaningful right context from mere capacity. The plug-in variant of \cref{sec:ablation}, in which the R2L pathway is trained on top of a frozen MDLM-adapted backbone, isolates this property: with only $5$B of additional tokens, the plug-in recovers $59$\% of the joint model's long-target gain over the causal baseline.

Bifocal is \emph{neither causal nor bidirectional}, but a distinct paradigm. The attention mechanism remains strictly causal, ensuring KV cache validity. The overall information flow is bidirectional: each position's representation incorporates both precise left context (via attention) and compressed right context (via the R2LM sidecar). This asymmetry is the source of both quality and efficiency: where bidirectional attention pays the same per-position cost in both directions, Bifocal delegates the right-context pathway to a lower-cost SSM, while preserving cacheable causal attention as the primary information channel. Our main contributions are:
\begin{itemize}[leftmargin=*]
    \item We introduce \textbf{Bifocal dLLMs}, a new diffusion-language-modeling paradigm built on the principle of \emph{asymmetric bidirectional context}. This is the first design that delivers continuous, position-aware right-side context while preserving standard causal prefix KV caching.
    \item We instantiate the paradigm as \textbf{R2LM}, a reverse state space sidecar attached to a pretrained autoregressive backbone via forward hooks with zero-initialized contribution. The same architecture supports both joint continued pretraining and a frozen-backbone plug-in mode, allowing integration into existing AR checkpoints without disturbing the pretrained model.
    \item We show empirically through controlled three-way continued pretraining of Qwen3-$1.7$B with $60$B tokens that R2LM achieves $2.4$ to $12.9\times$ higher throughput than a bidirectional dLLM while exceeding the causal baseline on most benchmarks and surpassing the bidirectional dLLM on the ALL average.
\end{itemize}

\section{Related work}
\label{sec:related}

\begin{table}[t]
  \caption{Taxonomy of dLLM paradigms. \textbf{Bifocal dLLM} is a paradigm whose attention pattern is \emph{Causal + $\underline{X}$}, where $\underline{X}$ denotes a right-context module attached as a residual to a causal Transformer backbone; the choice of $\underline{X}$ is an open architectural design space. We instantiate $\underline{X}$ as a reverse-direction Mamba SSM (\textbf{R2LM}, \cref{sec:architecture}).}
  \label{tab:paradigm}
  \centering
  \small
  \resizebox{\textwidth}{!}{%
  \begin{tabular}{llccc}
    \toprule
    \textbf{Paradigm} & \textbf{Method} & \textbf{Attention} & \textbf{Right Context} & \textbf{KV Cache} \\
    \midrule
    Bidirectional dLLM & LLaDA~{\scriptsize\citep{nie2025llada}} & Full bidir & Exact (full attention) & \ding{55} \\
    Causal dLLM & WeDLM~{\scriptsize\citep{liu2025wedlm}} & Causal & None & \ding{51} \\
    Block dLLM & BD3LM~{\scriptsize\citep{arriola2025block}} & Block-causal + bidir & Intra-block only & Block-level \\
    Hybrid AR-Diffusion & ReFusion~{\scriptsize\citep{refusion2026}} & Mixed causal + diffusion & Via diffusion component & Partial \\
    \midrule
    \textbf{Bifocal dLLM} & R2LM & \textbf{Causal + $\underline{X}$} & \textbf{Compressed} & \ding{51} \\
    \bottomrule
  \end{tabular}%
  }
\end{table}

\noindent\textbf{Diffusion models.}
Discrete diffusion adapts the continuous-time diffusion framework to categorical data by defining a Markov corruption process over discrete tokens. D3PM~\citep{austin2021structured} introduced transition-matrix-based discrete diffusion, and SEDD~\citep{lou2024sedd} cast discrete denoising as score estimation. The \emph{absorbing-state} variant, in which each corrupted position takes a dedicated \texttt{[MASK]} token, was simplified by Masked Diffusion Language Models (MDLM)~\citep{sahoo2024simple, shi2024simplified} and RADD~\citep{ou2025absorbing, ou2025absorbingdiscretediffusionslm} into a clean cross-entropy objective (\cref{eq:loss}) that is now standard. Discrete flow matching~\citep{gat2024discrete, campbell2024generative} extends the formulation to more general corruption processes. Earlier text-diffusion attempts such as DiffusionBERT~\citep{he2023diffusionbert} and SSDLM~\citep{han2023ssdlm} demonstrated feasibility at smaller scale. We adopt the absorbing-state MDLM objective throughout.

\noindent\textbf{Bidirectional and causal dLLMs.}
Existing dLLMs fall into four paradigms along the attention, context, and caching axes (\cref{tab:paradigm}). \emph{Bidirectional} dLLMs such as LLaDA~\citep{nie2025llada}, Dream~\citep{ye2025dream}, and DiffuLLaMA~\citep{gong2024diffullama} let every position attend to all others, maximising the available context but breaking prefix KV caching since each token's keys and values depend on yet-to-be-resolved positions. \emph{Causal} dLLMs such as WeDLM~\citep{liu2025wedlm} and CARD~\citep{ruan2026card} restore cacheability via topological reordering or masking schedules, but accept the right-context quality gap.

\noindent\textbf{Block and hybrid dLLMs.}
\emph{Block} diffusion models~\citep{arriola2025block, cheng2025sdar, wu2025fastdllmv2} apply bidirectional attention within fixed blocks and causal dependencies across blocks, recovering partial caching at the cost of bounding right context to a single block (typically $32$ to $128$ tokens). \emph{Hybrid AR-diffusion} models~\citep{refusion2026, liu2025tidar, sahoo2025esolms} combine AR and diffusion modes through dual-mode architectures or slot-level permutation, at the cost of architectural complexity. Across these four paradigms, the same trade-off persists: bidirectional attention breaks caching, while causal, block, or hybrid mechanisms restrict right-side information to a strict subset of the sequence. Unlike block diffusion~\citep{arriola2025block}, which restricts right context to a $32$ to $128$ token block, Bifocal supplies continuous, position-aware right context across the full sequence while maintaining standard causal prefix KV caching.

\noindent\textbf{State space models and hybrid architectures.}
Mamba~\citep{gu2024mamba} introduced a selective state space model (SSM) with input-dependent transition matrices, achieving linear-time sequence modeling with quality competitive with Transformers. Mamba-2~\citep{dao2024mamba2} formalized the duality between structured SSMs and attention. Jamba~\citep{lieber2024jamba} combines these by interleaving Mamba and attention layers, both operating left-to-right, within a single language model, using the SSM primarily to reduce attention cost on long contexts. Our R2LM architecture shares the broad idea of mixing attention and SSM, while differing in both direction and role: Jamba uses \emph{same-direction} (L2R + L2R) mixing, where the SSM provides capacity rather than novel information relative to causal attention; we use \emph{opposite-direction} (L2R attention + R2LM) mixing, where the SSM supplies right-context information that causal attention structurally cannot access. To our knowledge this is the first use of reverse SSM as an auxiliary right-context mechanism in diffusion language modeling.

\section{Method}
\label{sec:method}

We instantiate the Bifocal paradigm as \textbf{R2LM}, a discrete diffusion language model that pairs an unmodified causal Transformer backbone with a reverse-direction Mamba residual stream: the backbone delivers prefix-KV-cache-compatible left-context computation, and the reverse stream supplies a residual signal driven by right-context tokens.

\subsection{Preliminaries: Masked Diffusion Language Models}
\label{sec:prelim}

A masked diffusion language model (MDLM)~\citep{sahoo2024simple, nie2025llada} defines a generative model over discrete sequences via a forward corruption process and a learned reverse denoising process. The forward process is parameterised by a continuous time $t \in [0, 1]$ with survival probability $\alpha_t \in [0, 1]$ ($\alpha_0 = 1$, $\alpha_1 = 0$); each clean token $x_i$ in the sequence $\mathbf{x} = (x_1, \ldots, x_L)$ is independently kept with probability $\alpha_t$ and replaced by a special $\texttt{[MASK]}$ token otherwise:
\begin{equation}
q(z_t^i = m \mid x_i) \;=\; \alpha_t \cdot \mathbf{1}[m = x_i] \;+\; (1 - \alpha_t) \cdot \mathbf{1}[m = \texttt{[MASK]}].
\label{eq:forward}
\end{equation}
We use the linear schedule $\alpha_t = 1 - t$. Let $\mathcal{O}_t$ and $\mathcal{M}_t$ denote the observed and masked position sets at time $t$. A denoising network $p_\theta(x_i \mid \mathbf{z}_t)$ is trained by minimising the ELBO-weighted cross-entropy over masked positions~\citep{sahoo2024simple}:
\begin{equation}
\mathcal{L}(\theta) = \mathbb{E}_{t, \mathbf{x}, \mathbf{z}_t}\!\left[\, w(t) \sum_{i \in \mathcal{M}_t} -\log p_\theta(x_i \mid \mathbf{z}_t) \,\right],
\qquad w(t) = \frac{1}{1 - t}.
\label{eq:loss}
\end{equation}
At inference, generation starts fully masked and runs $T$ denoising steps; following~\citet{nie2025llada}, at each step we unmask the positions with highest max-softmax probability. All dLLMs compared in this paper train under \cref{eq:loss}; they differ only in the attention pattern through which $p_\theta$ exposes a masked position $i$ to the observed tokens $\mathbf{x}_{\mathcal{O}_t}$.

\subsection{Attention Patterns in Existing dLLMs}
\label{sec:two_paradigms}

For each masked position $i \in \mathcal{M}_t$ we partition the observed tokens by side: $\mathbf{x}_{\mathcal{O}_t}^{\leq i} = \{x_j : j \in \mathcal{O}_t,\, j \leq i\}$ to the left and $\mathbf{x}_{\mathcal{O}_t}^{>i} = \{x_j : j \in \mathcal{O}_t,\, j > i\}$ to the right of $i$.

\noindent\textbf{Bidirectional dLLM.}
The bidirectional dLLM~\citep{nie2025llada, ye2025dream} parameterises $p_\theta$ as a Transformer with full self-attention, so each position conditions on the entire noised sequence: $p_\theta(x_i \mid \mathbf{z}_t)$. Keys and values at every position depend on tokens at every other position, so prefix-KV caching across denoising steps is invalid. Each step recomputes attention at $O\!\big(B(P{+}G)^2\big)$ and FFN at $O\!\big(B(P{+}G) d^2\big)$ (batch $B$, prompt $P$, generation $G$).

\noindent\textbf{Causal dLLM.}
The strict-causal dLLM applies a causal mask, so position $i$ attends only to $j \leq i$ and the per-position score reduces to $p_\theta(x_i \mid \mathbf{z}_t^{\leq i})$. Prefix KV is computed once at prefill and reused across denoising steps; per-step attention drops to $O\!\big(B G (P{+}G)\big)$ in the $G \ll P$ regime and FFN to $O(B G d^2)$. The cost is that every masked position loses access to tokens at $j > i$. CARD~\citep{ruan2026card} uses specialised masking schedules but shares strict-causal attention.

\noindent\textbf{Block-wise interpolation.}
Block-diffusion paradigms~\citep{arriola2025block, cheng2025sdar, wu2025fastdllmv2, sdlm2025} apply bidirectional attention within fixed-size blocks of length $B_{\text{block}}$ and causal attention across blocks, recovering strict causal at $B_{\text{block}} = 1$ and full bidirectional at $B_{\text{block}} = L$. WeDLM~\citep{liu2025wedlm} additionally reorders observed tokens to the front of each block. Block diffusion caches prefixes at block granularity and limits each masked position's right context to its enclosing block.

All three patterns share one trade-off: right-context access costs prefix-cache validity, while restoring the cache discards right context.

\subsection{The Right-Context Gap and Its Residual Decomposition}
\label{sec:info_gap}

The quality gap between causal and bidirectional dLLMs has a clean information-theoretic form: it equals the conditional mutual information of the target with the right-side observed context. By the chain rule of mutual information, the per-position Bayes-optimal information gap between bidirectional and causal conditioning is
\begin{equation}
\Delta_i^t \;=\; H(x_i \mid \mathbf{x}_{\mathcal{O}_t}^{\leq i}) - H(x_i \mid \mathbf{x}_{\mathcal{O}_t}) \;=\; I\!\left(x_i;\, \mathbf{x}_{\mathcal{O}_t}^{>i} \,\middle|\, \mathbf{x}_{\mathcal{O}_t}^{\leq i}\right) \;\geq\; 0,
\label{eq:info_gap}
\end{equation}
and the expected total gap under the MDLM objective is $\mathbb{E}_t \big[\sum_{i \in \mathcal{M}_t} \Delta_i^t\big]$. At mask rate $\gamma$, the expected right-observed count at position $i$ is $(1 - \gamma)(L - i)$, so $\Delta_i^t$ is largest for early positions and vanishes at the rightmost masked position. The same right-context information also admits an additive log-posterior decomposition:

\begin{proposition}[Additive log-posterior decomposition]
\label{prop:score_decomp}
The Bayes-optimal log-posterior at masked position $i$ decomposes additively into a left-context term and a right-context correction:
\begin{equation}
\log p(x_i \mid \mathbf{x}_{\mathcal{O}_t}) = \underbrace{\log p(x_i \mid \mathbf{x}_{\mathcal{O}_t}^{\leq i})}_{\text{left-context term}} + \underbrace{\log \frac{p(x_i \mid \mathbf{x}_{\mathcal{O}_t})}{p(x_i \mid \mathbf{x}_{\mathcal{O}_t}^{\leq i})}}_{\Delta_i^{\textup{R2L}}:\;\text{right-context correction}}.
\label{eq:score_decomp}
\end{equation}
Taking the expectation of $\Delta_i^{\textup{R2L}}$ under the joint distribution of $(x_i, \mathbf{x}_{\mathcal{O}_t})$ recovers the per-position information gap:
\begin{equation}
\mathbb{E}_{(x_i, \mathbf{x}_{\mathcal{O}_t})}\!\left[\Delta_i^{\textup{R2L}}\right]
\;=\; \mathbb{E}_{\mathbf{x}_{\mathcal{O}_t}}\!\!\left[\,D_{\textup{KL}}\!\left(p(x_i \mid \mathbf{x}_{\mathcal{O}_t}) \;\|\; p(x_i \mid \mathbf{x}_{\mathcal{O}_t}^{\leq i})\right)\right]
\;=\; I\!\left(x_i;\, \mathbf{x}_{\mathcal{O}_t}^{>i} \,\middle|\, \mathbf{x}_{\mathcal{O}_t}^{\leq i}\right).
\label{eq:kl_mi}
\end{equation}
\end{proposition}

\begin{proof}[Proof sketch]
\cref{eq:score_decomp} is the identity $\log a = \log b + \log(a/b)$. For \cref{eq:kl_mi}, fix $\mathbf{x}_{\mathcal{O}_t}$, take the inner expectation under $p(x_i \mid \mathbf{x}_{\mathcal{O}_t})$ to recover $D_{\text{KL}}\!\big(p(\cdot \mid \mathbf{x}_{\mathcal{O}_t}) \,\|\, p(\cdot \mid \mathbf{x}_{\mathcal{O}_t}^{\leq i})\big)$, and outer-expect over $\mathbf{x}_{\mathcal{O}_t}$ to convert the conditional KL into a conditional mutual information.\end{proof}

\Cref{eq:score_decomp} motivates three architectural requirements: \ding{182}~\uline{preserve the causal pathway} so its forward pass delivers the left-context term and prefix-KV-cache reuse remains valid; \ding{183}~\uline{add a residual signal driven by tokens at $j > i$} to approximate $\Delta_i^{\textup{R2L}}$; \ding{184}~\uline{start from a zero residual} so the augmented model is bit-identical to the causal baseline at initialisation. LayerNorm and softmax break strict additivity at the logit level, so we treat \cref{eq:score_decomp} as an inductive bias rather than a strict equality.

\subsection{R2LM Architecture}
\label{sec:architecture}

\begin{figure*}[t]
\centering
\includegraphics[width=\textwidth]{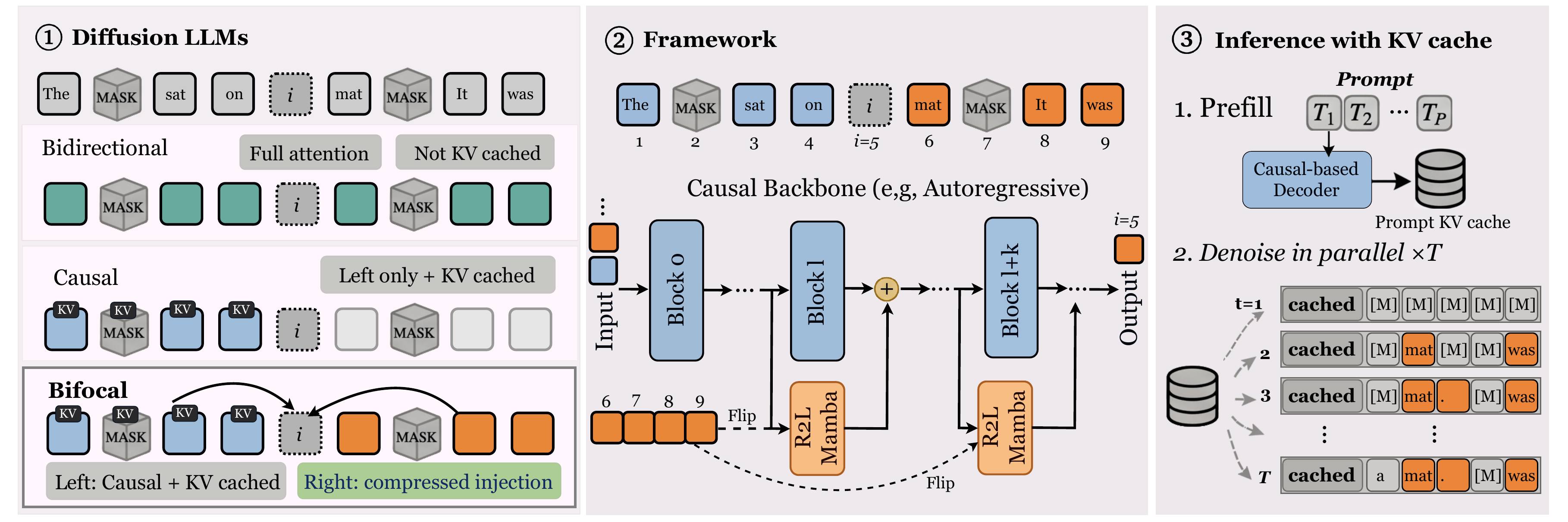}
\caption{R2LM architecture. A causal Transformer backbone (left-context, KV-cache compatible) is augmented with reverse Mamba SSM hooks attached after every $k$-th decoder layer. The R2L stream supplies the right-context correction $\Delta_i^{\textup{R2L}}$ via a $\tanh(s)$-gated residual with $s$ initialised to zero.}
\label{fig:architecture}
\end{figure*}

\Cref{eq:score_decomp} admits an architectural realization, R2LM (rightmost row of \cref{tab:paradigm}): an unmodified causal Transformer delivers the left-context term exactly, and a residual reverse-direction Mamba stream attached at a small number of decoder layers via forward hooks supplies the right-context correction $\Delta_i^{\textup{R2L}}$ (\cref{fig:architecture}).

\noindent\textbf{R2LM stream.}
Only the R2L direction supplies information that causal attention cannot access. Let $\mathbf{h} \in \mathbb{R}^{B \times L \times d}$ denote the hidden states at decoder layer $\ell$, where $B$ is the batch size, $L$ the sequence length, and $d$ the model dimension. The R2L stream produces an updated hidden state $\mathbf{h}^{+}$ via four steps: \ding{182}~\uline{sequence flip}: $\mathbf{h}_{\text{flip}} = \mathrm{flip}(\mathbf{h}, \mathrm{dim}{=}1)$, where $\mathrm{flip}(\cdot, \mathrm{dim}{=}1)$ reverses the sequence axis. The flip exposes right-side tokens to a left-to-right scan; without it, the SSM would summarise left context, which the causal backbone already provides. \ding{183}~\uline{selective scan}: $\mathbf{h}_{\text{mamba}} = \mathrm{Mamba}(\mathbf{h}_{\text{flip}})$, where $\mathrm{Mamba}(\cdot)$ is the Mamba-1 selective-scan block of~\citet{gu2024mamba}. We choose Mamba over alternatives (reverse self-attention, reverse RNN, linear attention) because its hidden state aggregates a position-aware summary of preceding tokens at linear time, and remains compatible with KV caching since the SSM scan does not consume causal-attention keys. \ding{184}~\uline{unflip and normalise}: $\mathbf{h}_{\text{R2L}} = \mathrm{LN}(\mathrm{flip}(\mathbf{h}_{\text{mamba}}, \mathrm{dim}{=}1))$, where $\mathrm{LN}$ is a layer normalisation local to this stream. The unflip restores the original sequence order so that $\mathbf{h}_{\text{R2L},i}$ summarises tokens at $j \geq i$, aligning the residual with backbone position $i$. The LayerNorm is necessary to bound the magnitude of the residual before it is added back to the backbone stream. \ding{185}~\uline{gated residual}: $\mathbf{h}^{+} = \mathbf{h} + \mathbf{h}_{\text{R2L}} \cdot \tanh(s)$, where $s \in \mathbb{R}$ is one learnable scalar gate per hooked layer, initialised at zero. The scalar (rather than vector) gate gives a single trackable per-layer quantity, and the $\tanh$ saturation bounds the residual contribution at $|\tanh(s)| \leq 1$, matching LayerScale and ControlNet. These four steps together realise requirement \ding{183} of the architectural list above, with the zero-initialised gate realising \ding{184}.

\subsection{Training and Inference}
\label{sec:plugin_train}

\noindent\textbf{Training modes.}
Both training modes minimise \cref{eq:loss} with a causal attention mask, sampling $t \sim \mathcal{U}(\epsilon, 1 - \epsilon)$ under the linear schedule and uniform random masking. The schedule is held fixed across variants so that an architectural comparison is not entangled with a schedule confound. In the \emph{joint} mode, both the backbone parameters and the R2L stream parameters are trained from a pretrained autoregressive checkpoint, with R2L weights initialised from scratch and scalar gates $\{s_\ell\}$ set to $0$; the augmented model is bit-identical to the causal dLLM at step $0$ and gradually opens the R2L pathway as the gates grow. In the \emph{plug-in} mode (\textbf{R2LM-PI}), the backbone is frozen at a separately trained causal dLLM checkpoint and only the R2L stream parameters (Mamba blocks, LayerNorms, and scalar gates) are updated, again with $s = 0$ at init. Plug-in mode serves both as a controlled probe that isolates the R2L stream from joint backbone training, and as an adapter recipe by which any pretrained causal dLLM hosts the R2L stream.

\noindent\textbf{Inference with prefix KV cache.}
At inference we follow standard cached parallel decoding (\cref{fig:architecture}, right panel): \ding{182}~\uline{prefill} forwards the prompt of length $P$ through the causal backbone and the R2L hooks once, caching prompt keys, values, and post-hook hidden states; \ding{183}~\uline{denoising} runs $T$ steps that process only the $G$ generation positions, reading prompt KV from cache and running the R2L scan on the generation block alone; \ding{184}~\uline{unmask}: at each step, positions with highest max-softmax probability are unmasked, following~\citet{nie2025llada}. The cache stays valid across denoising steps because cached hidden states already incorporate the post-hook residual stream from prefill. The per-step R2L cost is $O(BG\,d_{\text{ssm}}\,H)$, independent of $P$, while cached attention contributes $O(BG(P{+}G))$, recovering the throughput profile of cached causal serving while supplying the right-context residual that bidirectional dLLMs purchase by recomputing full attention each step.

\section{Experiments}
\label{sec:experiments}

We evaluate R2LM through a controlled three-way comparison against the bidirectional and causal dLLM endpoints under a matched $60$B-token continued-pretraining protocol on Qwen3-$1.7$B. We aim to answer the following research questions:
\begin{itemize}[leftmargin=*]
  \item \textbf{RQ1:} Does R2LM close the quality gap to bidirectional dLLMs at causal-level cost?
  \item \textbf{RQ2:} What inference-efficiency gain does KV-cache compatibility deliver?
  \item \textbf{RQ3:} Is the right-context pathway necessary, or do gains come from added parameters alone?
\end{itemize}

\subsection{Experimental Setup}
\label{sec:setup}

\noindent\textbf{Models.} All variants share the \texttt{Qwen3-1.7B}~\citep{yang2024qwen2} backbone ($28$ decoder layers, hidden dimension $d{=}2048$) and undergo continued pretraining under the MDLM objective of \cref{eq:loss}. The \emph{Causal dLLM} retains native causal attention. The \emph{Bidirectional dLLM} replaces the causal mask with a padding-only 4D mask. \emph{R2LM} (our Bifocal instantiation) adds a \emph{ReverseMambaLayer} hook after every $k{=}4$\textsuperscript{th} layer ($H{=}7$ hooks) with $d_\text{state}{=}16$, $d_\text{conv}{=}4$, expand $2$, and zero-init gates.

\noindent\textbf{Training.} All three runs share the same optimizer and data: AdamW at $1\!\times\!10^{-4}$, linear warmup over $500$ steps then cosine decay, weight decay $0.01$, per-device batch size $8$, no gradient accumulation, sequence length $4096$, \texttt{bfloat16} with gradient checkpointing, and full-shard FSDP across $32$ H100 GPUs ($4{\times}8$). The corpus is \texttt{HuggingFaceFW/fineweb-edu}~\citep{lozhkov2024fineweb} streamed with seed $42$, so all three runs see identical batches. Each model is trained for $60$B tokens.

\noindent\textbf{Plug-in.} \emph{R2LM-PI} starts from the causal dLLM checkpoint above with the backbone frozen, training only the seven \emph{ReverseMambaLayer} modules for $5$B further tokens at learning rate $5\!\times\!10^{-4}$ and warmup $100$ steps. This isolates the R2L pathway from joint backbone training and serves as an adapter recipe for any pretrained causal dLLM.

\noindent\textbf{Benchmarks.} We report on seven multiple-choice tasks following~\citep{sahoo2024simple, nie2025llada, ye2025dream}: ARC-Challenge ($25$-shot)~\citep{clark2018arc}, HellaSwag ($3$-shot)~\citep{zellers2019hellaswag}, PIQA ($0$-shot)~\citep{bisk2020piqa}, WinoGrande ($5$-shot)~\citep{sakaguchi2020winogrande}, OpenBookQA ($0$-shot)~\citep{mihaylov2018openbookqa}, BoolQ ($0$-shot)~\citep{clark2019boolq}, and MMLU ($5$-shot, $57$ subtasks)~\citep{hendrycks2021mmlu}. Throughput is swept on a single H100-$80$GB across $B \in \{1, 8\}$ and $P \in \{512, 1024, 2048, 4096\}$ with $G{=}128$ generated tokens and $T{=}32$ denoising steps.

\noindent\textbf{Evaluation.} Multiple-choice scoring uses \texttt{lm-evaluation-harness}~\citep{eval-harness} with $\text{mc\_num}{=}32$, \texttt{batch\_size}$=8$, and \texttt{max\_length}$=4096$. Causal and R2LM variants use the \texttt{CausalMDLMSampler} with KV-cache reuse; the bidirectional variant uses the \texttt{MDLMSampler}.

\subsection{Main Results: Multiple-Choice Benchmarks}
\label{sec:main_results}

\Cref{tab:main} reports multiple-choice accuracy across the seven benchmarks at $60$B tokens. We organize the columns into two blocks based on a single predictive property of each task: the length in tokens of the candidate continuation that the model is asked to score. Reasoning and commonsense tasks (ARC-C, HellaSwag, PIQA, WinoGrande, OpenBookQA) score multi-token candidates with median target length between $4$ and $20$~tokens. BoolQ (yes/no) and MMLU ($5$-shot, single letter answer) reduce to a single answer token.

\noindent\textbf{Long-target tasks.} On the five multi-token tasks the average accuracy of R2LM is $47.44$\%, compared with $43.90$\% for the causal dLLM ($+3.54$~pp) and $44.78$\% for the bidirectional dLLM ($+2.66$~pp). R2LM exceeds the causal baseline on all five (ARC-C $+0.7$, HellaSwag $+7.3$, PIQA $+3.4$, WinoGrande $+5.9$, OpenBookQA $+0.4$~pp). Versus the bidirectional dLLM the picture is mixed: R2LM beats bidirectional on PIQA ($+5.7$), WinoGrande ($+2.5$), and OpenBookQA ($+11.4$~pp), is roughly tied on HellaSwag ($-0.4$~pp), and trails bidirectional on ARC-C ($-5.9$~pp). On three of seven tasks R2LM therefore beats both endpoints; on a fourth it ties bidirectional while still beating causal; and the long-target average improves over both baselines. The result is consistent with \cref{eq:info_gap}: the value of right context grows with the amount of unmasked material to the right of each masked position, and the gain is largest where the candidate occupies many tokens.

\noindent\textbf{Single-token tasks.} R2LM continues to improve over the causal baseline on the two single-token tasks: $67.00$\% on BoolQ (causal $65.30$\%, $+1.70$~pp) and $29.80$\% on MMLU (causal $27.40$\%, $+2.40$~pp). The bidirectional dLLM is the strongest of the three on this regime ($68.40$\% / $34.90$\%), but R2LM closes most of the gap to bidir (BoolQ $-1.4$, MMLU $-5.1$~pp). The plug-in variant R2LM-PI extends the gain over causal further ($+2.0$~pp BoolQ, $+9.6$~pp MMLU) and surpasses bidir on MMLU by $+2.1$~pp.

\noindent\textbf{Headline averages.} Across all seven tasks R2LM averages $47.71$\%, exceeding both the causal baseline at $44.60$\% ($+3.11$~pp) and the bidirectional dLLM at $46.74$\% ($+0.97$~pp). R2LM-PI averages $\mathbf{47.76}$\% on ALL, marginally above joint R2LM and the strongest of the four models. R2LM-PI also leads the Short avg ($52.15$\% vs $51.65$\% for bidir) and MMLU ($37.00$\% vs $34.90$\% for bidir), achieving the strongest single-token performance overall. R2LM thus exceeds the causal baseline on most benchmarks and the bidirectional dLLM on the ALL average; R2LM-PI exceeds both endpoints on Long, Short, and ALL averages.

\noindent\textbf{R2LM-PI matches or exceeds joint R2LM across averages.} Despite training on a frozen causal backbone for only $5$B tokens (one-twelfth the data of joint training) and updating only the $185$M R2L parameters (one-tenth of the joint trainable count), R2LM-PI recovers $59$\% of the joint long-target gain and posts the highest ALL avg ($47.76$\%) and the highest Short avg ($52.15$\%) of the four models. The pattern indicates that the right-context pathway transfers cleanly to a pretrained MDLM backbone without joint training; we treat this as evidence of architectural decoupling between the causal backbone and the right-context residual (\cref{sec:ablation}).

\begin{table}[t]
\caption{Multiple-choice accuracy on seven benchmarks at $60$B-token CPT from Qwen3-$1.7$B. \texttt{Long}/\texttt{Short}/\texttt{ALL} columns average over $5$ multi-token / $2$ single-token / all $7$ tasks. Bold = best within column. Subscripts give $\Delta$ accuracy vs the \emph{Causal} baseline (red $\uparrow$ better, blue $\downarrow$ worse). $\uparrow$ higher is better.}
\label{tab:main}
\centering
\small
\setlength{\tabcolsep}{4pt}
\resizebox{\textwidth}{!}{%
\begin{NiceTabular}{@{}lccccc|cc|ccc@{}}[code-before = \rectanglecolor{metabg}{1-1}{3-11} \rectanglecolor{metabg}{5-1}{5-11} \rectanglecolor{metabg}{8-1}{8-11}]
\toprule
 & \multicolumn{5}{c|}{\textbf{Long-target tasks (multi-token candidate)}} & \multicolumn{2}{c|}{\textbf{Short-target}} & \multicolumn{3}{c}{\textbf{Average}} \\
\textbf{Model} & ARC-C & HellaSw & PIQA & WG & OBQA & BoolQ & MMLU & Long & Short & ALL \\
 & \scriptsize $25$-sh & \scriptsize $3$-sh & \scriptsize $0$-sh & \scriptsize $5$-sh & \scriptsize $0$-sh & \scriptsize $0$-sh & \scriptsize $5$-sh & & & \\
\midrule
Bidirectional dLLM       & $\bm{45.30}$\accBetter{6.60} & $\bm{43.40}$\accBetter{7.70} & $61.50$\accWorse{2.30} & $\underline{57.70}$\accBetter{3.40} & $16.00$\accWorse{11.00} & $\bm{68.40}$\accBetter{3.10} & $\underline{34.90}$\accBetter{7.50} & $44.78$\accBetter{0.88} & $\underline{51.65}$\accBetter{5.30} & $46.74$\accBetter{2.14} \\
Causal dLLM              & $38.70$ & $35.70$ & $63.80$ & $54.30$ & $\underline{27.00}$ & $65.30$ & $27.40$ & $43.90$ & $46.35$ & $44.60$ \\
\midrule
\multicolumn{11}{@{}l}{\emph{\textbf{Bifocal dLLMs} (this work)}} \\
\quad\textbf{R2LM}    & $39.40$\accBetter{0.70} & $\underline{43.00}$\accBetter{7.30} & $\bm{67.20}$\accBetter{3.40} & $\bm{60.20}$\accBetter{5.90} & $\bm{27.40}$\accBetter{0.40} & $67.00$\accBetter{1.70} & $29.80$\accBetter{2.40} & $\bm{47.44}$\accBetter{3.54} & $48.40$\accBetter{2.05} & $\underline{47.71}$\accBetter{3.11} \\
\quad\textbf{R2LM-PI (plug-in)}  & $\underline{40.90}$\accBetter{2.20} & $39.50$\accBetter{3.80} & $\underline{64.60}$\accBetter{0.80} & $59.40$\accBetter{5.10} & $25.60$\accWorse{1.40} & $\underline{67.30}$\accBetter{2.00} & $\bm{37.00}$\accBetter{9.60} & $\underline{46.00}$\accBetter{2.10} & $\bm{52.15}$\accBetter{5.80} & $\bm{47.76}$\accBetter{3.16} \\
\bottomrule
\end{NiceTabular}%
}
\end{table}

\subsection{Inference Efficiency}
\label{sec:speed}

\noindent\textbf{Throughput.} \cref{fig:efficiency} and \cref{tab:speed} report serving~\citep{liang2026generative} on a single H100-$80$GB at $32$ denoising steps and $128$ generated tokens. The bidirectional dLLM's throughput collapses from $483$ to $53$~tok/s as $P$ grows from $512$ to $4096$ at $B{=}8$, since every denoising step recomputes attention over the full $(P{+}G)$ sequence. R2LM's KV-cached path stays nearly flat ($1154\!\to\!683$~tok/s), widening the gap from $2.4\times$ at $P{=}512$ to $\mathbf{12.9\times}$ at $P{=}4096$. R2L overhead versus the causal dLLM is $7$ to $16$\%, shrinking with $P$ because the SSM cost is bounded by the generation block.

\noindent\textbf{Memory.} Peak GPU memory at $B{=}1$ grows near-linearly for the bidirectional dLLM ($4.5\!\to\!18.5$~GB at $P\in\{1024, 16384\}$) and sublinearly for R2LM ($4.5\!\to\!12.7$~GB). At $P{=}16384$ R2LM saves $5.9$~GB ($32$\%) over the bidirectional dLLM while staying within $0.4$~GB of causal. These results show the R2L SSM preserves the causal memory profile while supplying right-side context.

\begin{figure*}[t]
\centering
\includegraphics[width=\textwidth]{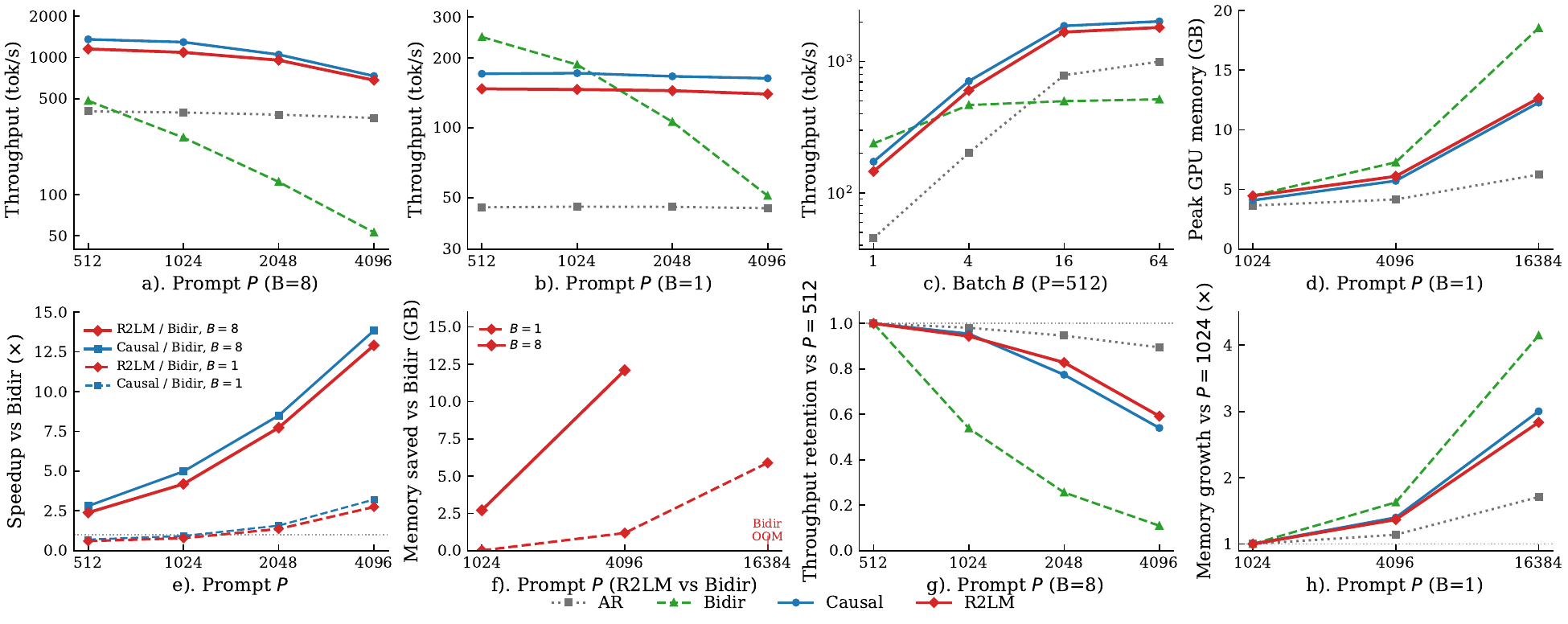}
\caption{Inference efficiency on H100-$80$GB ($32$ denoising steps, $128$ generated tokens). \textbf{(a, b)}~Throughput vs prompt length $P$ at $B{=}8$ and $B{=}1$. \textbf{(c)}~Throughput vs batch $B$ at $P{=}512$; AR clipped at $10^3$ (true $2939$ tok/s at $B{=}64$). \textbf{(d)}~Speedup over Bidir for the KV-cached variants. \textbf{(e, f)}~Peak memory and memory saved vs Bidir at $B{=}1$. \textbf{(g, h)}~Throughput retention vs $P{=}512$ at $B{=}8$; memory growth vs $P{=}1024$ at $B{=}1$.}
\label{fig:efficiency}
\end{figure*}

\begin{table*}[t]
\caption{Throughput (tok/s) on a single H100-$80$GB at $32$ denoising steps and $128$ generated tokens. Subscripts give the speedup over \emph{Bidir} at the same $(B,P)$ (red $\uparrow$ faster, blue $\downarrow$ slower). Bold = best within column, underlined = second-best. $\uparrow$ higher is better.}
\label{tab:speed}
\centering
\small
\setlength{\tabcolsep}{3pt}
\resizebox{\textwidth}{!}{%
\begin{NiceTabular}{@{}l|cccc|cccc@{}}[code-before = \rectanglecolor{metabg}{1-1}{2-9} \rectanglecolor{metabg}{4-1}{4-9} \rectanglecolor{metabg}{7-1}{7-9}]
\toprule
\multirow{2}{*}{\textbf{Model}} & \multicolumn{4}{c|}{\textbf{$B{=}1$ (single-request)}} & \multicolumn{4}{c}{\textbf{$B{=}8$ (saturated, deployment regime)}} \\
\cmidrule(lr){2-5} \cmidrule(lr){6-9}
& $P{=}512$ & $P{=}1024$ & $P{=}2048$ & $P{=}4096$ & $P{=}512$ & $P{=}1024$ & $P{=}2048$ & $P{=}4096$ \\
\midrule
Qwen3-1.7B (AR)        & $45.4$\spdWorse{0.18} & $45.7$\spdWorse{0.24} & $45.6$\spdWorse{0.43} & $45.0$\spdWorse{0.88} & $404$\spdWorse{0.84} & $397$\spdBetter{1.53} & $383$\spdBetter{3.09} & $362$\spdBetter{6.83} \\
Causal dLLM            & $\underline{170.9}$\spdWorse{0.69} & $\underline{171.9}$\spdWorse{0.92} & $\bm{166.6}$\spdBetter{1.57} & $\bm{163.4}$\spdBetter{3.21} & $\bm{1356}$\spdBetter{2.81} & $\bm{1295}$\spdBetter{4.98} & $\bm{1050}$\spdBetter{8.47} & $\bm{732}$\spdBetter{13.8} \\
Bidir dLLM             & $\bm{246.2}$ & $\bm{187.0}$ & $106.0$ & $50.9$ & $483$ & $260$ & $124$ & $53$ \\
\midrule
\multicolumn{9}{@{}l}{\emph{\textbf{Bifocal dLLMs} (this work)}} \\
\quad\textbf{R2LM}  & $147.1$\spdWorse{0.60} & $146.2$\spdWorse{0.78} & $\underline{144.5}$\spdBetter{1.36} & $\underline{139.7}$\spdBetter{2.74} & $\underline{1154}$\spdBetter{2.39} & $\underline{1089}$\spdBetter{4.19} & $\underline{956}$\spdBetter{7.71} & $\underline{683}$\spdBetter{12.9} \\
\bottomrule
\end{NiceTabular}%
}
\end{table*}

\subsection{Ablation Studies}
\label{sec:ablation}

To answer RQ3 we isolate the R2L mechanism with two parameter-matched controls trained jointly with the causal backbone on $1$B tokens (\cref{tab:ablation}): an \emph{MLP Adapter} that lacks any directional structure, and an \emph{L2R Mamba} that runs in the wrong direction. We report WikiText-103 perplexity (via MC log-likelihood, following~\citet{sahoo2024simple}) because at the $1$B-token CPT scale, downstream MC accuracy is too undertrained to discriminate between variants.

\noindent\textbf{Direction matters; parameters alone do not.} The MLP Adapter matches R2LM's parameter count ($185$M) but provides no directional structure: it improves training loss only marginally ($4.66\!\to\!4.63$) and \emph{worsens} WikiText perplexity by $2\times$ ($178\!\to\!356$), ruling out the ``more parameters'' hypothesis. The L2R Mamba uses the same Mamba block but scans left-to-right (redundant with causal attention) and ends at a worse training loss ($4.84$) and the worst perplexity in the table ($374$): the wrong direction is actively harmful, not merely useless. Only R2LM's reversed scan recovers the bidirectional upper bound, closing $90.2$\% of the causal-to-bidirectional training-loss gap and reducing WikiText perplexity by $4.2\times$.

\begin{table}[t]
\caption{Mechanism ablation under matched $1$B-token joint CPT on Qwen3-$1.7$B. WikiText-103 perplexity is via MC log-likelihood ($\text{mc\_num}{=}32$) on a $100$-chunk subset; train loss is the running mean over the last $100$ steps. $\downarrow$ lower is better.}
\label{tab:ablation}
\centering
\small
\setlength{\tabcolsep}{6pt}
\begin{NiceTabular}{@{}lccc@{}}[code-before = \rectanglecolor{metabg}{1-1}{1-4} \rectanglecolor{metabg}{5-1}{5-4}]
\toprule
\textbf{Variant} & \textbf{Extra params} & \textbf{Train loss} $\downarrow$ & \textbf{WikiText-103 PPL} $\downarrow$ \\
\midrule
Causal dLLM (baseline)                    & $0$    & $4.66$          & $178$ \\
\quad + MLP Adapter (no direction)        & $185$M & $4.63$          & $356$ \\
\quad + L2R Mamba (original direction)       & $185$M & $4.84$          & $374$ \\
\quad + R2L Mamba (\textbf{R2LM, ours})   & $185$M & $\bm{3.18}$ & $\bm{42}$ \\
Bidirectional dLLM          & $0$    & $3.02$          & $30$ \\
\bottomrule
\end{NiceTabular}
\end{table}

\section{Conclusion}
\label{sec:conclusion}

We presented Bifocal dLLMs, a paradigm for discrete diffusion language models that pairs causal-level inference efficiency with right-context access through asymmetric bidirectional context. The R2LM instantiation combines a causal Transformer backbone with a reverse Mamba pathway injected via forward hooks, delivering KV-cache-compatible inference and exceeding the causal baseline on most benchmarks. The plug-in variant R2LM-PI, trained on a frozen causal backbone with one-twelfth the data and one-tenth the trainable parameters of joint training, posts the highest ALL average among the four models, indicating that the right-context pathway transfers cleanly to pretrained MDLM checkpoints.

\noindent\textbf{Limitations.} The model and training budget ($1.7$B parameters, $60$B tokens) are below the regime in which dLLMs and AR baselines establish their strongest published numbers; results may evolve at scale. The R2L stream adds $10.8$\% backbone parameters, between LoRA-style adapters and full bidirectional attention; lower-rank or shared-projection variants are open. Single-request inference at $B{=}1$ is slightly slower than a pure causal dLLM by the R2L overhead.

\bibliographystyle{assets/plainnat}
\bibliography{references}

\end{document}